\documentclass[aoas,preprint]{imsart}

\RequirePackage[OT1]{fontenc}
\RequirePackage{amsthm,amsmath}
\RequirePackage{natbib}
\RequirePackage[colorlinks,citecolor=blue,urlcolor=blue]{hyperref}

\RequirePackage{amsmath,amsthm,amssymb,bbm}
\RequirePackage{graphicx}
\RequirePackage{booktabs}

\startlocaldefs
\numberwithin{equation}{section}
\theoremstyle{plain}

\newtheorem{lemma}{Lemma}[section]
\endlocaldefs

\let \hat \widehat

\begin{document}

\begin{frontmatter}
\title{Functional regression for quasar spectra\thanksref{T1}}
\runtitle{Functional regression for quasar spectra}
\thankstext{T1}{A six pages preliminary extended abstract version of this work was accepted among the conference proceedings of the Third International Workshop on Functional and Operatorial Statistics  (IWFOS2014, Stresa, Italy; June 19-21, 2014).}

\begin{aug}
\author{\fnms{Mattia} \snm{Ciollaro}\ead[label=e1]{ciollaro@cmu.edu}},
\author{\fnms{Jessi} \snm{Cisewski}\ead[label=e2]{cisewski@stat.cmu.edu}},
\author{\fnms{Peter E.} \snm{Freeman}\ead[label=e3]{pfreeman@stat.cmu.edu}},
\author{\fnms{Christopher R.} \snm{Genovese}\thanksref{s1,s2} \ead[label=e4]{genovese@stat.cmu.edu}},
\author{\fnms{Jing} \snm{Lei}\ead[label=e5]{jinglei@andrew.cmu.edu}},
\author{\fnms{Ross} \snm{O'Connell}\ead[label=e6]{rcoconne@andrew.cmu.edu}}
\and
\author{\fnms{Larry} \snm{Wasserman}\thanksref{s1}\ead[label=e7]{larry@stat.cmu.edu}}
\ead[label=u1,url]{http://www.stat.cmu.edu}
\ead[label=u2,url]{http://www.cmu.edu/physics}

\thankstext{s1}{Supported in part by National Science Foundation Grant NSF-DMS-1208354.}
\thankstext{s2}{Supported in part by Department of Energy Grant DE-FOA-0000918.}

\runauthor{M. Ciollaro et al.}

\affiliation{Carnegie Mellon University}

\address{M. Ciollaro\\
J. Cisewski\\
P. E. Freeman\\
C. R. Genovese\\
J. Lei\\
L. Wasserman\\
\text{}\\
Department of Statistics\\
Carnegie Mellon University\\
5000 Forbes Avenue\\
Pittsburgh (PA)\\
15213 USA\\
\printead{e1}\\
\phantom{E-mail:\ }\printead*{e2}\\
\phantom{E-mail:\ }\printead*{e3}\\
\phantom{E-mail:\ }\printead*{e4}\\
\phantom{E-mail:\ }\printead*{e5}\\
\phantom{E-mail:\ }\printead*{e7}\\
\printead{u1}
}

\address{R. O'Connell\\
\text{}\\
Department of Physics\\
Carnegie Mellon University\\
5000 Forbes Avenue\\
Pittsburgh (PA)\\
15213 USA\\
\printead{e6}\\
\printead{u2}
}
\end{aug}

\begin{abstract}
The {\em Lyman-$\alpha$ forest} is a portion of the observed light spectrum of distant galactic nuclei which allows us to probe remote regions of the Universe that are otherwise inaccessible. The observed Lyman-$\alpha$ forest of a quasar light spectrum can be modeled as a noisy realization of a smooth curve that is affected by a `damping effect' which occurs whenever the light emitted by the quasar travels through regions of the Universe with higher matter concentration. To decode the information conveyed by the Lyman-$\alpha$ forest about the matter distribution, we must be able to separate the smooth {\em continuum} from the noise and the contribution of the damping effect in the quasar light spectra. To predict the continuum in the Lyman-$\alpha$ forest, we use a nonparametric functional regression model in which both the response and the predictor variable (the smooth part of the damping-free portion of the spectrum) are function-valued random variables. We demonstrate that the proposed method accurately predicts the unobservable continuum in the Lyman-$\alpha$ forest both on simulated spectra and real spectra. Also, we introduce distribution-free prediction bands for the nonparametric functional regression model that have finite sample guarantees. These prediction bands, together with bootstrap-based confidence bands for the projection of the mean continuum on a fixed number of principal components, allow us to assess the degree of uncertainty in the model predictions.
\end{abstract}

\begin{keyword}
\kwd{Nonparametric functional regression}
\kwd{Functional data analysis}
\kwd{Prediction}
\kwd{Quasar spectra}
\kwd{Lyman-$\alpha$ forest}
\end{keyword}

\end{frontmatter}

\section{Introduction}

\begin{figure}[h!]
	\begin{center}
		\includegraphics[width=\columnwidth]{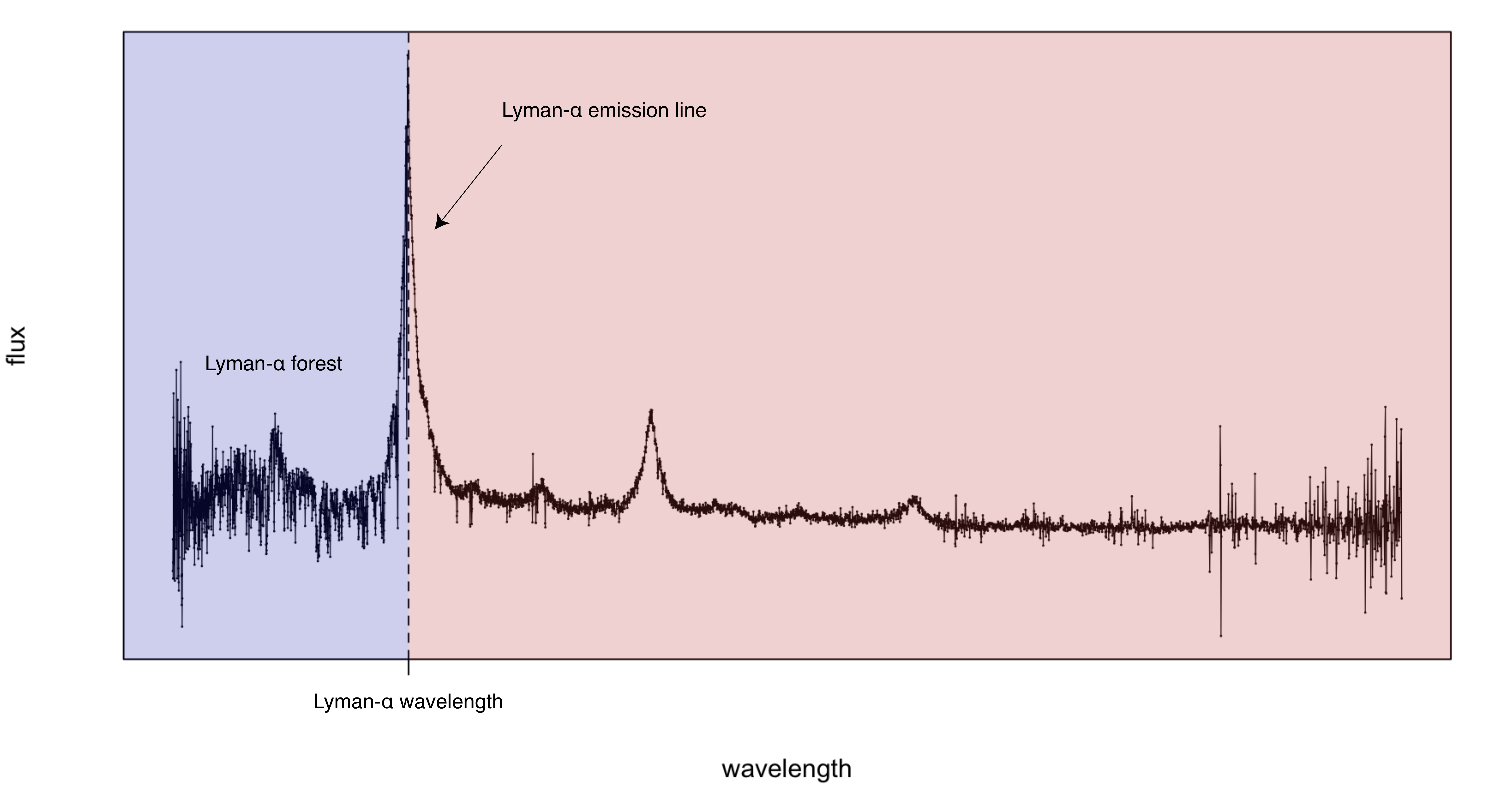}
		\caption{Light spectrum of a quasar. To the left of the Lyman-$\alpha$ line, the flux is damped at a dense set of wavelengths and the unabsorbed flux continuum is not clearly recognizable (blue side of the figure). To the right of the Lyman-$\alpha$ line, the flux continuum is easily guessed (red portion of the figure). The spectrum is perturbed by heteroskedastic noise with higher variance at the extremes of the observed wavelength range.}
		\label{fig:quasarhigh}
	\end{center}
\end{figure}

Technological advances over
the last quarter-century have allowed astronomers to collect
data of unprecedented richness and scope, gathered from
previously inaccessible regions of space.
By exploiting the information about the spatial distribution
of the neutral hydrogen atoms contained in these data,
cosmologists are placing ever-tighter constraints on
the fundamental parameters
governing the Universe's structure and evolution.
Today, quasar spectra are one of the only means for
probing the distribution of neutral hydrogen and they allow us
to investigate the Universe at distances that are far beyond the reach of other available methods.

Quasars are luminous distant galactic nuclei that appear as point-like sources of light in our sky.
Many properties of a quasar, and of the regions of the Universe through which its light passes on its way to us, can be inferred from its {\em light spectrum}.
This is a curve that relates the light's intensity, or {\em flux}, to its wavelength.
Figure~\ref{fig:quasarhigh} shows an example of a quasar light spectrum. For wavelengths greater than what is called the Lyman-$\alpha$ wavelength, the observed light spectrum $f_{\text{obs}}$ can be well modeled as a smooth {\em continuum} $f$ plus noise:

\begin{equation*}
	f_{\text{obs}}(\lambda) = f(\lambda) + \text{noise}(\lambda).
\end{equation*}

\noindent For wavelengths below the Lyman-$\alpha$ wavelength -- a region of the spectrum known as the {\em Lyman-$\alpha$ forest} -- a third contribution comes from a nonsmooth random absorption effect that originates whenever the light emitted by the quasar travels through regions of the Universe that are richer in neutral hydrogen:

\begin{equation}
	\label{eq:absorptionmodeleasy}
		f_{\text{obs}}(\lambda) = \text{absorption}(\lambda) \cdot f(\lambda) + \text{noise}(\lambda).
\end{equation}

We want to decode the information in the Lyman-$\alpha$ forest to learn about the distribution of the neutral hydrogen in otherwise unreachable regions of the Universe. To do this, however, we need to separate the smooth part of the quasar spectrum, which we call the {\em unabsorbed flux continuum} (UFC), from the noise and (especially) from the absorption component. The task of separating the three components is not an easy one because there is no immediately available information about the UFC in the Lyman-$\alpha$ forest (which we label UFC$\alpha$). However, the smooth part of the spectrum above the Lyman-$\alpha$ wavelength, henceforth abbreviated with UFC+, could be informative about the structure of UFC$\alpha$. In fact, these two curves belong to the same object, and one may be used to predict the other. For a particular spectrum, let $X$ represent UFC+ (where the quasar spectrum is absorption-free) and let $Y$ denote UFC$\alpha$. We develop a nonparametric functional regression model to estimate the regression operator

\begin{equation}
	\label{eq:firstregressionoperator}
		r(X)(\cdot) = E(Y | X )(\cdot),
\end{equation}

\noindent where

\begin{equation*}
		\lambda \mapsto r(X)(\lambda)
\end{equation*}

\noindent is a curve that has the same domain as $Y$. Then, we use the estimate $\hat r$ to predict $Y^*$ of a new quasar spectrum from $X^*$ according to

\begin{equation*}
	\hat Y^*(\lambda) = \hat r(X^*)(\lambda).
\end{equation*}

Reliable predictions of UFC$\alpha$ have important implications in cosmological studies. The detection of Baryon Acoustic Oscillations (BAOs) is only one example of the many goals which require good quality predictions of the UFC. BAOs are `typical' distances that characterize the geometry of our Universe. They provide a valuable ruler to measure the separation between pairs of distant objects and to study large scale structures. Our ability to detect BAOs using Lyman-$\alpha$ forest data relies on the accuracy of the predicted UFC$\alpha$. Thus, while on the one hand the prediction of UFC$\alpha$ has a high scientific payoff, on the other hand it also constitutes a challenging statistical task for three main reasons. First, the target quasar spectra are those of {\em high redshift} (i.e. very distant) quasars. Because of the strength of the absorption effect
in the Lyman-$\alpha$ forest of their spectra, the UFC is indiscernible in the observed spectrum of high redshift quasars. One then necessarily has to use low redshift quasar spectra (in which the absorption effect is absent or mild) to fit the regression model, and the goodness of the predictions is then conditional, at least in part, on the redshift invariance of the regression operator $r$. Secondly, the amplitude of the spectra in the Lyman-$\alpha$ forest is characterized by a high degree of variability. For some spectra, the variability in the amplitude represents a natural limit in the ability of statistical models to predict UFC$\alpha$. Finally, there are no natural candidate predictors for UFC$\alpha$ other than the absorption-free part of the spectrum, UFC+.
Despite these challenges, we demonstrate that nonparametric functional regression provides a natural framework for this prediction problem and that the proposed nonparametric functional regression model produces satisfying predictions of UFC$\alpha$.

Our work also introduces distribution-free prediction bands for the nonparametric functional regression model with a function-valued predictor and a function-valued response. Prediction bands for nonparametric functional models are not yet well developed in the Functional Data Analysis literature: we contribute by providing a straightforward method to obtain prediction bands with finite sample coverage guarantees. We use these bands to assess the uncertainty of the model predictions in the Lyman-$\alpha$ continuum prediction problem. Thus, by combining some recent advances in nonparametric functional regression modeling with distribution-free prediction bands, we provide a complete and statistically sound methodology for the prediction of UFC$\alpha$ and the assessment of the uncertainty in the predictions.

The remainder of the paper is organized as follows. The next section provides more scientific details and background on quasars, the Lyman-$\alpha$ forest, and the motivation behind the need of statistical methods to predict UFC$\alpha$. Section \ref{sec:nonparametricfunctionalmodeling} presents the nonparametric estimator of the regression operator $r$ of equation \eqref{eq:firstregressionoperator}. In Section \ref{sec:conformalsets}, we derive distribution-free prediction bands for the UFC while Section \ref{sec:confidencebands} describes the implementation of a {\em wild bootstrap} procedure which can be used to generate pseudo-confidence bands for the mean of UFC$\alpha$. The prediction bands and the confidence bands allow us to visually evaluate the degree of uncertainty of the model predictions. Sections \ref{sec:simulationstudy} and \ref{sec:realspectraapplication} are devoted to describing the application of the nonparametric functional regression model to a set of realistic simulated spectra and to two sets of real quasar spectra. Finally, in Section \ref{sec:conclusions} we summarize the results of our study and give some perspective on future research directions stemming from this analysis.

\section{Quasars and the Lyman-$\alpha$ forest}
\label{sec:quasarsandlyaforest}
Shortly after the Big Bang, the Universe 
was sufficiently hot that hydrogen gas existed in an ionized 
state, meaning that the negatively charged electrons within the gas 
were not bound to positively charged 
protons.
After 400,000 years, the
Universe had cooled sufficiently that electrons and protons recombined
and hydrogen became neutral. But over the next billion years, as stars and 
galaxies formed and bathed the intergalactic medium (IGM) in ultraviolet (UV)
radiation, the Universe became reionized. Reionization was not
instantaneous; within a volume of space, one can envision ionized regions
as initially small bubbles around the hottest stars that over time expanded 
and merged, with regions of neutral gas steadily shrinking
and eventually evaporating.

As we look further and further from the Earth, we look deeper and deeper
into the past, and with sufficiently powerful telescopes we can
probe remote regions of the Universe that contain 
neutral hydrogen.
Neutral hydrogen overdensities are far too faint for us to directly observe, but
we may infer their presence when they are backlit by distant luminous objects 
such as quasars. A quasar is the central core, or nucleus, of a galaxy that becomes
transiently ultraluminous when matter falls into a supermassive black hole
at a sufficiently high rate. As
matter falls, it converts up to 
40\% of its mass to photons that stream away from the black hole; the 
conversion of mass to radiant energy is so efficient that
a quasar's luminosity can dwarf that of the surrounding galaxy.\footnote{
Hence the name quasar, or quasi-stellar object:
the first observed quasars appeared on photographic plates
as point-like objects that were
indistinguishable from the stars of the Milky Way.}

Photons interact with neutral hydrogen at specific wavelengths which are typically
measured in \textit{\aa ngstr\"{o}ms}, denoted \AA\, (1 \AA ~$= 10^{-10}$ m).
For instance, a UV photon with a wavelength of 1216 \AA\, can cause an electron
to gain energy and jump to an excited state (more precisely, the electron ascends
from the first to the second orbital);
since energy must be conserved, the photon is absorbed. 
This specific interaction leads to the formation of
the so-called {\em Lyman-$\alpha$ absorption line}, a reduction in
the number of photons with wavelength 1216 \AA
relative to the number at other nearby wavelengths.
To detect regions containing neutral hydrogen, we rely on the fact that
as UV photons stream from quasars to the Earth, they are {\em redshifted}:
the Universe continuously expands over time, with a commensurate increase
in the wavelengths of the photons. A photon's redshift is defined as
\begin{eqnarray}
z = \frac{\lambda_{\rm obs}}{\lambda_{\rm emit}} - 1 \,, \nonumber
\end{eqnarray}
where $\lambda_{\rm emit}$ and $\lambda_{\rm obs}$
are the wavelengths when the photon is emitted and observed, respectively.
Cosmologists use redshift as a directly observable proxy for distance: the
farther a photon travels, the higher the redshift. Thus, regions
that are richer of neutral hydrogen along a quasar's line of sight
lying at different distances/redshifts from the quasar, absorb photons 
that have $\lambda_{\rm obs}$ = 1216 \AA, but
different values of $\lambda_{\rm emit}$.
In simpler terms, the light spectrum of a quasar is shifted towards higher wavelengths as a result of the Universe's expansion (with spectra of more distant quasars being
subject to a larger shift) and the dips corresponding to the absorption of light by neutral hydrogen are thus spread over a broad range of wavelengths rather than being coincident at the 1216 \AA\, wavelength.
The resulting series of dips in
the quasar's spectrum is dubbed the Lyman-$\alpha$ forest
(see Figure~\ref{fig:quasarhigh} and, e.g., \citealp{rauch1998annualreview} and \citealp{rauch1998lymanalphaencyclopedia}).
By convention, cosmologists concentrate on that part of the Lyman-$\alpha$ forest
extending downwards from 1216 \AA~to
1026 \AA, the wavelength of the Lyman-$\beta$ transition (where electrons
ascend from hydrogen's first orbital to its third orbital).

The absorption component of equation \eqref{eq:absorptionmodeleasy} can be expressed as $e^{-\tau(\lambda_{\rm emit})}$, where 
$\tau(\lambda_{\rm emit}) \geq 0$ is the {\em optical depth} (essentially, the local density) of the neutral hydrogen 
at the quasar rest-frame wavelength $\lambda_{\rm emit}$. This gives a more specific version of equation
\eqref{eq:absorptionmodeleasy},

\begin{equation}
	\label{eq:fluxmodel}
		f_{\rm obs}(\lambda_{\rm emit}) = e^{-\tau (\lambda_{\rm emit})} f(\lambda_{\rm emit}) + \epsilon (\lambda_{\rm emit}) \,,
\end{equation}
where $\epsilon$ is a zero-mean noise process (capturing, for instance, instrumental noise) and $f$ is
the smooth, but unobserved, quasar {\em continuum}. From equation \eqref{eq:fluxmodel}, it follows that
a larger $\tau(\lambda_{\rm emit})$ implies a deeper dip at wavelength $\lambda_{\rm emit}$
in the Lyman-$\alpha$ forest.
However, because the Universe is locally ionized (i.e.
neutral hydrogen is observable only if one considers sufficiently large neighborhoods of an arbitrary location in
the Universe), $\tau \approx 0$ in the Lyman-$\alpha$ forest portion of low redshift spectra. To simplify notation,
we henceforth drop the subscript from $\lambda_{\text{emit}}$.

Nonparametric smoothing of 
Lyman-$\alpha$ forest data allows direct estimates of
the product $e^{-\tau(\lambda)}f(\lambda)$ of equation \eqref{eq:fluxmodel}, but this product is by itself uninteresting;
analysts ultimately wish to estimate $\tau$ (or some transformation of it) because $\tau$ is the parameter
that links the absorption in the Lyman-$\alpha$ forest to the distribution of neutral hydrogen. In fact, from a scientific standpoint, the
smooth UFC $f$ can be considered a nuisance parameter.
The task of estimating the optical depth $\tau$ has led to the development
of algorithms by which one may estimate $f$ (and thus $\tau$)
in the Lyman-$\alpha$ forest regime
using the unabsorbed spectral flux
at wavelengths above $\approx$ 1300 \AA, based on principal
components analysis (see, e.g., \citealp{suzuki2005continuum}, \citealp{paris2011principal}, and \citealp{lee2012mean}).

The estimation of the relative flux absorption 

\begin{equation}
	\label{eq:relativefluxabsorption}
		\delta(\lambda)=\frac{f_{\text{obs}}(\lambda)-f(\lambda)}{f(\lambda)} \approx e^{-\tau(\lambda)}-1\,
\end{equation}

\noindent is a case where good predictions of the UFC $f$ are crucial. The correlation function of $\delta$ across different lines of sight can be used to detect BAOs, which manifest as periodic patterns in the separation between overdense regions of our Universe. However, to estimate the correlation function of the relative flux absorption $\delta$ for high redshift quasars across different lines of sight, one needs a good prediction of $f$, which is unobservable because of the absorption effect in the Lyman-$\alpha$ forest caused by intervening neutral hydrogen. The reader can refer to \cite{eisenstein2005darkmauscript} and \cite{eisenstein2005dark} for an introduction to BAOs, \cite{refId0} and \cite{slosar2011lyman} for examples of BAOs detection using the correlation function of $\delta$ and the correlation function of the galaxy distribution using data from the Baryon Oscillation Spectroscopic Survey (BOSS), and \cite{dawson2013baryon} for a description of the BOSS survey. 

Quasar spectra are high dimensional objects. For instance,
\cite{suzuki2005continuum} observe that, for low redshift quasars, about ten eigenspectra
are needed to capture all of the physically relevant features of their spectra. Because of the high dimensional and complex nature of Lyman-$\alpha$ forest data,
the study of Lyman-$\alpha$ forest is challenging and yet very appealing, especially for nonparametric statistical modeling.
\cite{cisewskiLyaIGM} use nonparametric methods to demonstrate how
Lyman-$\alpha$ forest data can be used to construct a 3D map of the
high redshift IGM that is otherwise not accessible.

In the following sections, we demonstrate that the problem of predicting UFC$\alpha$ (i.e. the continuum $f$ in the Lyman-$\alpha$ forest) can be addressed by estimating in a nonparametric fashion the conditional expectation $r(X)=E(Y|X)$ of a pair of function-valued random variables $(X,Y)$ from a sample of $n$ independent pairs $(X_i,Y_i)_{i=1}^n$ distributed as $(X,Y)$, and then by using the estimated regression operator $\hat r$ to predict $Y$ from $X$. We apply this procedure on a set of realistic mock spectra and on two sets of real spectra, and we find that one can use $\hat r$ to obtain accurate predictions of UFC$\alpha$. Our predictions are complemented with pseudo-confidence bands for the mean of UFC$\alpha$ and prediction bands with  finite sample guarantees. The prediction bands allow us to visually assess the variability (especially in terms of amplitude) of UFC$\alpha$. Our contribution is, to the best of our knowledge, the first application of modern Functional Data Analysis techniques to the analysis of Lyman-$\alpha$ forest data.

\section{Nonparametric functional modeling of quasar spectra}
\label{sec:nonparametricfunctionalmodeling}
We build on the work of \cite{lee2012mean}, \cite{paris2011principal}, and \cite{suzuki2005continuum}, and we use the $\lambda \geq 1300 $ \AA\, part of the absorption-free portion of the quasar spectrum UFC+ to predict UFC$\alpha$ in the 1050-1185 \AA\, range of the Lyman-$\alpha$ forest. In particular, we estimate the conditional expectation

\begin{equation*}
		r(X)(\lambda) = E(Y|X)(\lambda)\, \text{ for } \underline{\lambda} \leq \lambda \leq \overline{\lambda}
\end{equation*}

\noindent where, using the same notation of Section \ref{sec:quasarsandlyaforest}, $X$ and $Y$ are smooth function-valued random variables here representing, respectively, UFC+ in the $\lambda \geq 1300$ \AA\, wavelength region and UFC$\alpha$ in the $[\underline{\lambda},\overline{\lambda}] = [1050,1185]\, \text{\AA}$ wavelength region. In this setting, $X$ is a function-valued predictor for the function-valued response $Y$, and $r$ is an operator that maps functions into functions. More precisely, we have

\begin{equation*}
	r: X \mapsto r(X)
\end{equation*}

\noindent and for a fixed predictor $X$

\begin{equation*}
	r(X): \lambda \mapsto r(X)(\lambda),
\end{equation*}

\noindent meaning that $r(X)$ is a smooth curve that has the same domain of the response, $Y$.

Suppose that we observe a random sample of $n$ absorption-free and noiseless spectra with $\tau=0$ in the forest (recall from the previous section that $\tau \approx 0$ for low redshift spectra) and the corresponding i.i.d. pairs $\left(X_i, Y_i \right)_{i=1}^n$; then, for a new pair $(X,Y)$ for which $Y$ is not observable due to absorption (this is the case for high redshift spectra of ongoing sky surveys such as BOSS), we can predict $Y$ by means of the nonparametric functional regression estimator introduced by \cite{ferraty2012regression}

\begin{equation}
	\label{eq:model}
		\hat Y(\lambda)=\hat r(X)(\lambda)=\frac{\sum_{i=1}^n K\left(\frac{d(X_i,X)}{h}\right) \, Y_i(\lambda)}{\sum_{i=1}^n K\left(\frac{d(X_i,X)}{h}\right)}
\end{equation}

\noindent for $\underline{\lambda} \leq \lambda \leq \overline{\lambda}$.
Here, $d$ is a semimetric on the space of the predictor functions, the bandwidth $h$ is a $k$-nearest-neighbors smoothing parameter and $K$ is an asymmetric kernel function with compact support (in particular, we take $K$ to be a quadratic kernel supported in $[0,1]$). The bandwidth $h=h(\kappa,X)$ is a function of the number of nearest neighbors $\kappa$ and of the predictor $X$ which determines the width of the kernel $K$ at $X$. The number of nearest neighbors $\kappa$ is such that the number of sample predictor curves that satisfy $d(X_i,X) \leq h(\kappa,X)$ is exactly equal to $\kappa$. The role of the bandwidth parameter $h$ is to control the size of the neighborhood of $X$ that is used for the estimation of $r(X)$: choosing a kernel that is supported on $[0,1]$ implies that only those $Y_i$'s whose associated $X_i$ satisfies $d(X_i,X) \leq h(\kappa,X)$ contribute to equation \eqref{eq:model}. If $\kappa$ (and hence $h$) is too large, $\hat r(X)$ has high bias, whereas if $\kappa$ (and hence $h$) is too small $\hat r(X)$ has large variance. In practice, $\kappa$ is selected by cross-validation in the attempt of optimally balancing the bias and the variance of $\hat r$.

The function-valued pairs $(X_i,Y_i)_{i=1}^n$ used to fit the model and the new predictors $X$ are obtained by smoothing the appropriate portions of the observed noisy spectra.
We remark that, from a scientific standpoint, fitting the model on low redshift quasar spectra and then performing the predictions on high redshift quasar spectra entails the implicit assumption that the unknown conditional expectation operator $r$ is redshift invariant.

Notice that once the sample $(X_i,Y_i)_{i=1}^n$ is fixed and the optimal smoothing parameter $\kappa$ is chosen, the prediction of UFC$\alpha$ for new spectra is very fast to compute: only the kernel weights need to be updated in equation \eqref{eq:model} as $X$ varies.

\section{Conformal prediction bands}
\label{sec:conformalsets}
Equation \eqref{eq:model} provides a point prediction corresponding to the estimated conditional expectation of the function-valued random variable $Y$ when the corresponding function-valued predictor is $X$. In a predictive setting, besides a point prediction, it is also useful to have some measure of the uncertainty associated to the prediction rule. The uncertainty in the prediction necessarily comprises both the variability due to the estimation of the regression operator $r$ and the intrinsic variability of the response variable.

We complement the prediction rule \eqref{eq:model} with distribution-free prediction regions for the function-valued response $Y$ that have guaranteed finite sample coverage for any level $\alpha \in (0,1)$.
We use the \textit{conformal prediction} method of \cite{vovk2009line} and, in particular, we adapt the \textit{Inductive Conformal Predictor} algorithm to fit the nonparametric functional regression framework.

For a new pair of function values random variables $(X,Y) \sim P$, independent of the i.i.d. sample $(X_i,Y_i)_{i=1}^n \sim P$, we say that the random set $C_n(x)=C(X_1,\dots,X_n,x) \subseteq \mathcal{Y}$ is a \textit{marginally valid} $1-\alpha$ conformal prediction set if

\begin{equation*}
	\mathbb{P}\left(Y \in C_n(X) \right) \geq 1-\alpha 
\end{equation*}

\noindent for any $n \geq 1$, for any $\alpha \in (0,1)$ and for any $P$ where $\mathbb{P}=P^{(n+1)}=P \times P \times \dots \times P$ denotes the $(n+1)$-fold product probability measure based on $P$.

Let $\mathcal{X}$ and $\mathcal{Y}$ denote the function spaces to which the $X$'s and the $Y$'s belong, respectively. In the following, we denote $L=[\underline{\lambda},\overline{\lambda}]$, we assume that any $Y \in \mathcal{Y}$ is such that $P_Y(\sup_{\lambda \in L} |Y(\lambda)| < \infty)=1$, and we fix $\| y \|_{\mathcal{Y}}=\sup_{\lambda \in L} |y(\lambda)|$, which ensures that the conformal prediction region for $Y$ which we now described is band-shaped. Next, we consider the \textit{conformity score}

\begin{equation*}
	g(x,y)=-\| y-\hat r(x) \|_\mathcal{Y}=-\sup_{\lambda \in L} \left| y(\lambda) - \hat r(x)(\lambda) \right|\,,
\end{equation*}

\noindent which measures the conformity of the pair $(x,y)$ with the observed sample $(X_i,Y_i)_{i=1}^n$, i.e. how similar is $(x,y)$ with the sample. Consider now the following algorithm.\\

\fbox{\parbox{\textwidth}{
\begin{center}
	\textbf{Functional inductive conformal predictor}\\
\end{center}

\noindent \textbf{Input:} sample $\mathcal{S}=\{(X_1,Y_1), \dots, (X_n,Y_n)\}$, $X$; coverage level $1-\alpha$\\
\textbf{Output:} $C_n$, a level $1-\alpha$ marginally valid band for $Y$

\begin{enumerate}
	\item set $n_1=\lfloor \frac{n}{2} \rfloor$
	\item sample without replacement $n_1$ pairs from $\mathcal{S}$ to define a subsample $\mathcal{S}_1$
	\item define $\mathcal{S}_2=\mathcal{S} \setminus \mathcal{S}_1$
	\item estimate $\hat r_1$ on $\mathcal{S}_1$ and let $g_1(x,y)=-\sup_{\lambda \in L} |y(\lambda)-\hat r_1(x)(\lambda)|$
	\item compute the $\frac{\lfloor (n_2+1)\alpha\rfloor}{n_2+1}$ quantile $q$ of $g_1$ on $\mathcal{S}_2$
	\item output $C_n(X)=\{ y \in \mathcal{Y} : g_1(X,y) \geq q\}$\,.
\end{enumerate}
}}
\text{}\\

\noindent The intuition behind the algorithm is that, under the null hypothesis that $(X,Y)$ is distributed as the $(X_i,Y_i)$'s, the ranks of 
the conformity scores computed on $\mathcal{S}_2 \cup \{(X,Y)\}$ are uniformly distributed on $\left\{1, 2, \dots, n_2+1 \right\}$ therefore, with probability at least $1-\alpha$, $g(X,Y) \geq q$.

\begin{lemma}
The prediction set $C_n(X)$ for $Y$ of the algorithm described above is marginally valid.
\end{lemma}
\begin{proof}
After changing $X_1,\dots X_n$ into $(X_1,Y_1), \dots, (X_n,Y_n)$, Lemma 2.1 of \cite{lei2013conformal} guarantees that $\mathbb{P}((X,Y) \in \bar C_n) \geq 1-\alpha$, where $\bar C_n=\{(x,y) \in \mathcal{Y} : g_1(x,y) \geq q\}$. Since $C_n(X)$ is an $x$-slice of $\bar C_n$, the result follows.
\end{proof}

\section{Pointwise confidence bands for the regression operator}
\label{sec:confidencebands}
One may also be interested in constructing a confidence set $A_n$ for the unknown regression operator $r(X)(\cdot)=E(Y|X)(\cdot)$, corresponding to the (conditional) mean of UFC$\alpha$ in $L=[\underline{\lambda},\overline{\lambda}]$. In the context of nonparametric functional regression it is hard to derive exact confidence bands for $r(X)$. Even resampling-based confidence regions can be difficult to justify theoretically.

In this section, we discuss the implementation of a \textit{wild bootstrap} procedure described by \cite{ferraty2012regression} to obtain a confidence band for the projection of $r(X)$ on the subspace spanned by the first $m$ principal components of the observed sample of response functions $(Y_i)_{i=1}^n$.
The resulting set $A_n$ satisfies the property

\begin{equation*}
		P^{(n)} \left( \tilde r^m(X)(\lambda) \in A_n, \forall  \lambda \in L \right) \gtrsim 1-\alpha\,,
\end{equation*}

\noindent for any $P$, for $\alpha \in (0,1)$ and for a fixed $x \in \mathcal{X}$, where $\tilde r^m(x)$ is the projection of $r(x)$ onto the subspace spanned by the first $m$ principal components of $(Y_i)_{i=1}^n$. The symbol $\gtrsim$ indicates that the inequality only holds  for large enough $n$.

The wild bootstrap is a bootstrap procedure that uses a perturbed version of the estimated model residuals $e^*_{i,b}$ to generate a bootstrap sample of response functions $Y_{i,b}^*$ according to

\begin{equation*}
	Y_{i,b}^* = \hat r(X_i) + e^*_{i,b}
\end{equation*}

\noindent in order to obtain a bootstrap sample of spectra $(X_i,Y^*_{i,b})_{i=1}^n$. The noisy bootstrap residuals are generated so that they have mean 0 and have the same variance as the original residuals. This way, the bootstrap response functions $Y_{i,b}^*$'s have mean $\hat r(X_i)$ and the same variance of the $Y_i$'s under the probability induced by the resampling scheme. The wild bootstrap procedure developed by \cite{ferraty2012regression} works as follows.\\

\fbox{\parbox{\textwidth}{
\begin{center}
	\textbf{Functional wild bootstrap}\\
\end{center}

\noindent \textbf{Input:} sample $\mathcal{S}=\{(X_1,Y_1), \dots, (X_n,Y_n)\}$; first $m$ principal components of $(Y_i)_{i=1}^n$, $(\hat \phi_j)_{j=1}^m$; evaluation point $x$; confidence level $1-\alpha$; number of bootstrap replicates $B$\\
\textbf{Output:} $A_n$, a level $1-\alpha$ confidence band for $\tilde r^m(x)$

\begin{enumerate}
	\item estimate $\hat r(x)$ on $\mathcal{S}$ (i.e. choose the optimal bandwidth $h$)
	\item define $e_i = Y_i - \hat r(X_i)$ for $i \in \{1, \dots, n\}$
	\item for $b \in \{1, \dots, B\}$ repeat
		\begin{itemize}
			\item sample with replacement $n$ functions $(e_{i,b})_{i=1}^n$ from $(e_i)_{i=1}^n$
			\item sample $n$ independent draws $(V_i)_{i=1}^n$ such that $E(V_i)=0$ and $E(V_i^2)=1$ for $ i \in \{1, \dots, n\}$
			\item define $e_{i,b}^*=e_iV_i$ for $i \in \{1, \dots, n\}$
			\item define $Y^*_{i,b} = \hat r(X_i) + e^*_{i,b}$
			\item compute $\hat r_{b}(x)=\frac{\sum_{i=1}^n K\left(\frac{d(X_i,x)}{h}\right) \, Y^*_{i,b}(\lambda)}{\sum_{i=1}^n K\left(\frac{d(X_i,x)}{h}\right)}$
			\item project $\hat r_{b}(x)$ onto $(\hat \phi_j)_{j=1}^m$ and get $\hat r_{b,j}(x)=\langle \hat r_{b}(x), \hat \phi_j \rangle$ for $j \in \{1, \dots, m\}$
		\end{itemize}
	\item for each $j \in \{1, \dots, m\}$ compute the $\alpha/(2m)$ and the $1-\alpha/(2m)$ quantiles of the bootstrap distribution of $\hat r_{b,j}(x)$, say $q_{j,\alpha/(2m)}$ and $q_{j,1-\alpha/(2m)}$ 
	\item $A_n=\{\sum_{j=1}^m a_j \hat \phi_j : q_{j,\alpha/(2m)} \leq a_j \leq q_{j,1-\alpha/(2m)} \quad \forall j \in \{1,\dots, m\} \}$
\end{enumerate}
}}
\text{}\\

\noindent In the implementation of the wild bootstrap we set $V_i \sim 0.1(5+ \sqrt{5}) \delta_{(1-\sqrt{5})/2} + 0.1(5- \sqrt{5}) \delta_{(1+\sqrt{5})/2}$ which guarantees that $E(V_i)=0$ and $E(V_i^2)=E(V_i^3)=1$ and $\langle \hat r_{b}(x), \hat \phi_j \rangle = \int_{L} \hat r_{b}(x)(\lambda) \cdot \hat \phi_j(\lambda)\, d\lambda$ (which is in practice approximated by computing a Riemann sum on a finite grid of $\lambda$'s).

\section{Simulation study}
\label{sec:simulationstudy}
We fit the model described in Section \ref{sec:nonparametricfunctionalmodeling} on a set of $n=100$ mock quasar spectra and perform the prediction of UFC$\alpha$ in the 1050-1185 \AA\, wavelength range of a set of another set of 100 mock spectra that are generated according to the same recipe, but are not used at any time for model fitting.
The mock spectra are simulated using the $N=10$ eigenspectra $\xi_j$ and the mean flux component $\mu$ provided by \cite{suzuki2005continuum} (which were derived from a set of real low redshift quasar spectra normalized at the $\approx1280$ \AA\, wavelength) according to

\begin{equation}
	\label{eq:mockspectra}
		f_{\text{mock}}(\lambda)=\mu(\lambda)+\sum_{j=1}^{N} \omega_j \xi_j(\lambda) + \eta(\lambda)\,,
\end{equation}

\noindent where the $\omega_j$'s are independent $\mathcal{N}(0,\lambda_j)$ draws and $\lambda_j$ is the eigenvalue associated to the $j$-th eigenspectra, $\xi_j$. Finally, $\eta(\lambda) \sim \mathcal{N}(0,\sigma^2(\lambda))$ where the flux standard deviation function $\sigma$ is provided by \cite{suzuki2005continuum}. We do not simulate absorption caused by neutral hydrogen in the mock spectra as this is unnecessary for the evaluation of the predictions of the nonparametric regression model: in fact, since we are trying to predict UFC$\alpha$, the model predictions must be evaluated through a comparison with the true simulated UFC$\alpha$.

The noisy mock spectra $f_{\text{mock}}$ are smoothed using quadratic local polynomials on the two wavelength ranges of interest in order to derive the smooth predictor and response functions. The local polynomial smoothing is performed using the R function \verb+loess+ and the optimal smoothing parameter is chosen using 2-fold cross-validation. In equation \eqref{eq:model}, we found that setting the semimetric $d$ to be the metric induced by the $L_2$ norm was preferable over the Sobolev semimetrics based on the first, the second or higher order derivatives of the predictor functions, as the latter appeared to be less sensitive to amplitude and slope variability among the mock spectra and produced less accurate predictions. The regression operator $r$ is estimated using the \verb+ffunopare.knn.gcv+ function of the R code accompanying \cite{ferraty2012regression}.

The top panels of Figure~\ref{fig:mockresults} depict two successful predictions of the flux continuum on the mock spectra together with marginally valid conformal prediction bands. The bottom panels depict the projection of $\hat r(X)$ onto the first 5 principal components computed on the Lyman-$\alpha$ portion of the spectrum of the sample of 100 mock spectra, along with wild bootstrap confidence bands computed as described in Section \ref{sec:confidencebands}. In the two top panels, we observe that the shape of the continuum in the 1050-1185 \AA\, range is recovered with high accuracy. We point out, however, that for some mock spectra the amplitude of the flux continuum is not accurately predicted (see Figure~\ref{fig:mockspectraamplitude}). This is observed through a visual comparison of the predicted flux continuum and the true flux continuum of the mock spectra, and reflected in the prediction accuracy measure of Figure~\ref{fig:mockpredictionerror}. The model does not entirely capture the variability of the flux continuum amplitude in the Lyman-$\alpha$ forest portion of the spectrum, although this limitation is not unique to the functional modeling approach: other authors (see for instance \citealp{lee2012mean} and \citealp{paris2011principal}) encountered the same problem on real data when performing the prediction using different methodologies, and a possible solution consists in employing some ad hoc adjustment or post-processing of the model predictions in order to make them consistent with external measurements of the mean flux in the forest (see for instance the mean flux regulation procedure proposed by \citealp{lee2012mean}). This kind of adjustments or mean flux regulations, which we do not attempt in this work, often appear to positively influence the accuracy of the predictions, but they come at the cost of introducing a somewhat unpleasant post-processing step in the prediction procedure. It is not surprising that this difficulty is encountered on our simulated spectra, as their eigenspectra were in fact derived from real data. \cite{desjacques2007distribution}, \cite{lee2012mean}, \cite{paris2011principal}, \cite{zheng1997composite}, and others suggest that the amplitude variability in the Lyman-$\alpha$ forest portion of the spectra might be related to a break in the underlying power law of the UFC which occurs at $\lambda \approx 1200-1300$ \AA.

Notice that the width of the confidence bands in Figures \ref{fig:mockresults} and \ref{fig:mockspectraamplitude} is reasonably small, suggesting that most of the uncertainty in the prediction of the amplitude of UFC$\alpha$ comes from the intrinsic variability of the data rather than from the uncertainty in the estimation of the regression operator $r$.

\begin{figure}[h!]
	\begin{center}
		\includegraphics[width=\columnwidth]{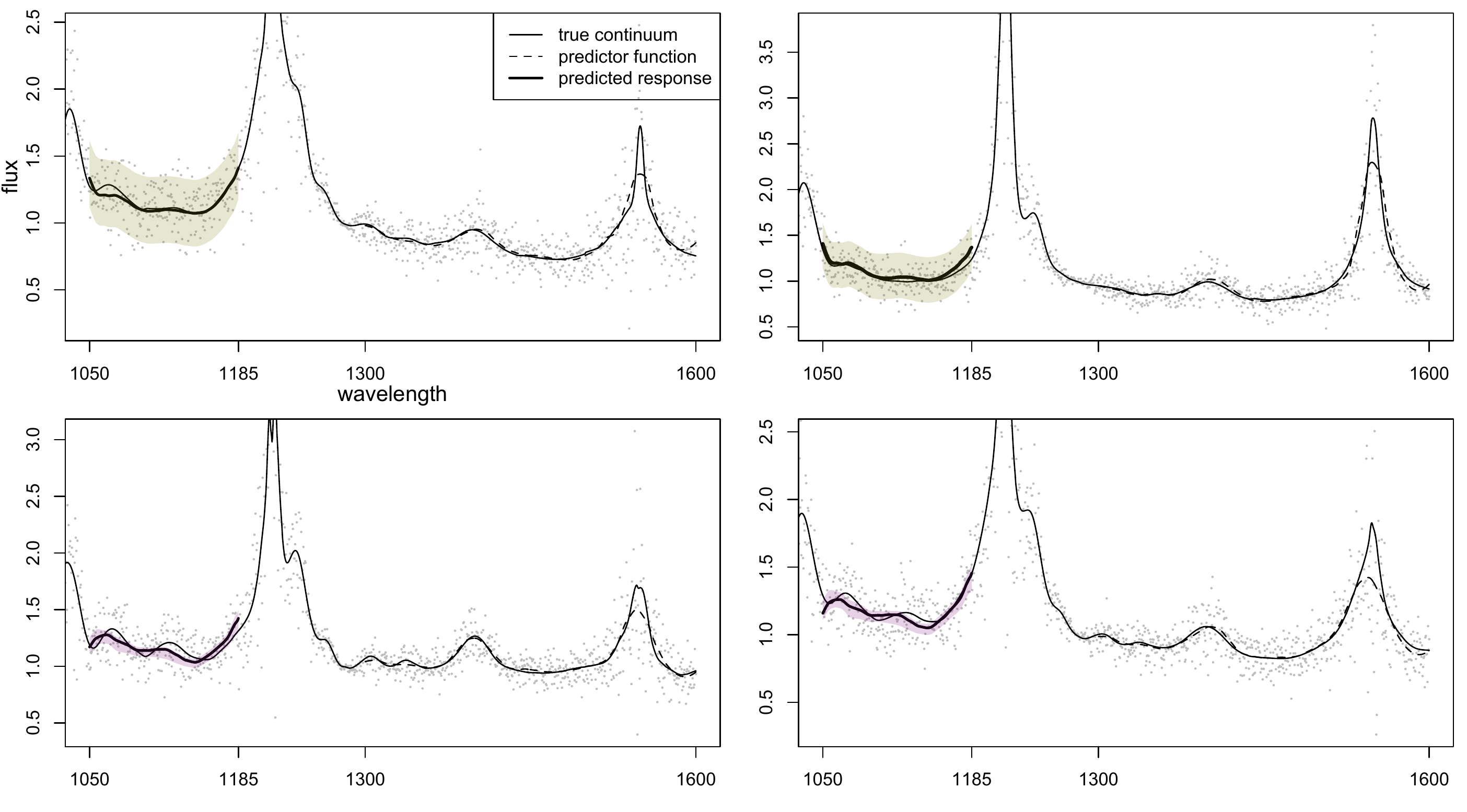}
		\caption{\underline{Top panels:} successful predictions of the Lyman-$\alpha$ flux continuum on mock spectra using the nonparametric functional regression model and the prediction equation \eqref{eq:model}. The thin continuous line is the true mock continuum (which we pretend to be unknown), the dashed line on the right is the locally quadratic functional representation of the true unabsorbed continuum derived from the noisy mock spectrum and the bold continuous line on the left is the predicted unabsorbed flux continuum. The gray bands are 90\% marginally valid conformal prediction bands. \underline{Bottom panels:} projection of $\hat r(X)$ onto the first $m=5$ principal components ($\approx 98\%$ of total variance) of the smooth response functions of the mock spectra $(Y_i)_{i=1}^{100}$. The pink bands are 90\% wild bootstrap confidence bands for the projection of the unknown regression operator $r(X)$ ($B=500$).}
		\label{fig:mockresults}
	\end{center}
\end{figure}

\begin{figure}[h!]
	\begin{center}
		\includegraphics[width=\columnwidth]{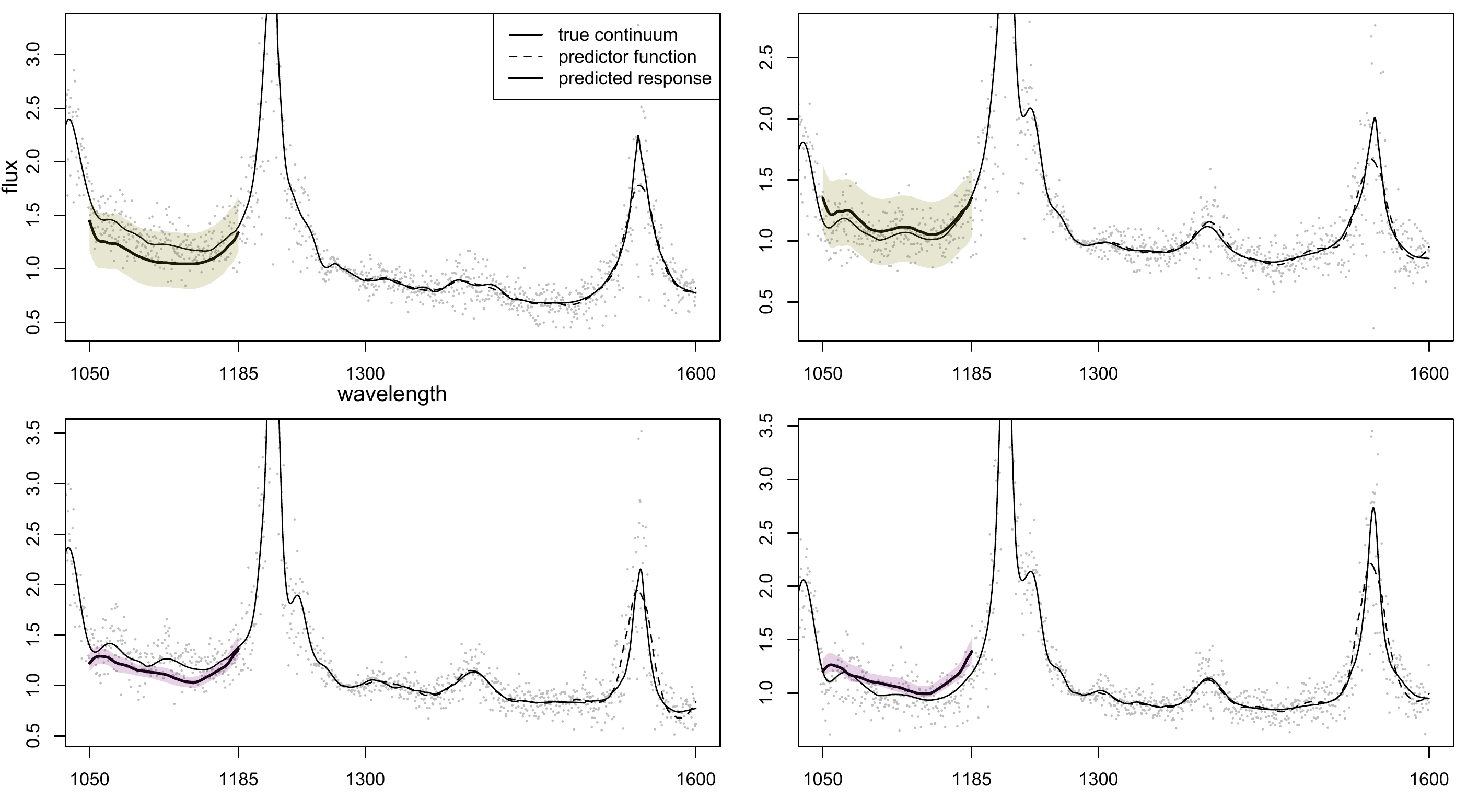}
		\caption{\underline{Top panels:} predictions of the Lyman-$\alpha$ flux continuum on mock spectra using the nonparametric functional regression model and the prediction equation \eqref{eq:model}. For these spectra, the amplitude of the flux continuum in the Lyman-$\alpha$ forest is not accurately predicted, although the overall shape of the predicted continuum matches the shape of the true continuum (which we pretend to be unknown). The thin continuous line is the true mock continuum, the dashed line on the right is the locally quadratic functional representation of the true continuum derived from the noisy mock spectrum and the bold continuous line on the left is the predicted flux continuum. The gray bands are 90\% marginally valid conformal prediction bands. \underline{Bottom panels:} projection of $\hat r(X)$ onto the first $m=5$ principal components ($\approx 98\%$ of total variance) of the smooth response functions of the mock spectra $(Y_i)_{i=1}^{100}$. The pink bands are 90\% wild bootstrap confidence bands for the projection of the unknown regression operator $r(X)$ onto the first $m=5$ principal components ($B=500$).}
		\label{fig:mockspectraamplitude}
	\end{center}
\end{figure}

In order to measure the quality of the predictions, we look at the relative prediction error on the mock spectra at each wavelength $\lambda$ in the 1050-1185 \AA\, portion of the spectrum

\begin{equation*}
	e(\lambda) = \frac{\left| \hat Y(\lambda) - f_{\text{mock}}^{\text{cont.}}(\lambda) \right|}{f_{\text{mock}}^{\text{cont.}}(\lambda)}\,,
\end{equation*}

\noindent where the true mock continua $f_{\text{mock}}^{\text{cont.}}$ are known by construction and correspond to the $\mu(\lambda)+\sum_{j=1}^{N} \omega_j \xi_j(\lambda)$ summand of equation \eqref{eq:mockspectra}. Figure~\ref{fig:mockpredictionerror} summarizes the relative prediction error on the mock spectra. We also look at the plain prediction error

\begin{equation*}
	u(\lambda) = \hat Y(\lambda) - f_{\text{mock}}^{\text{cont.}}(\lambda)
\end{equation*}

\noindent whose empirical counterpart on the 100 mock spectra is summarized in Figure \ref{fig:mockerror}. The empirical average of the relative prediction error over the $[1050,1185]$ \AA\, wavelength range is approximately 5.7\% for the 100 mock spectra for which we performed the prediction. The first and the third quantiles, averaged over the target wavelength range, are respectively 2.6\% and 8.1\% approximately. Ideally, one wishes that both $e$ and $u$ are close to 0 for each spectrum. However, even for moderately large $e$'s, if $u$ is close to 0 on an aggregate sample of spectra (i.e. the prediction errors $u$'s zero out on average on an aggregate sample of spectra), one is still able to
detect BAOs signals in the correlation function of $\delta$. From Figure \ref{fig:mockerror} it appears that $u$ is very close to 0 on average on the 100 mock spectra.

\begin{figure}[h]
	\begin{center}
		\includegraphics[width=\columnwidth]{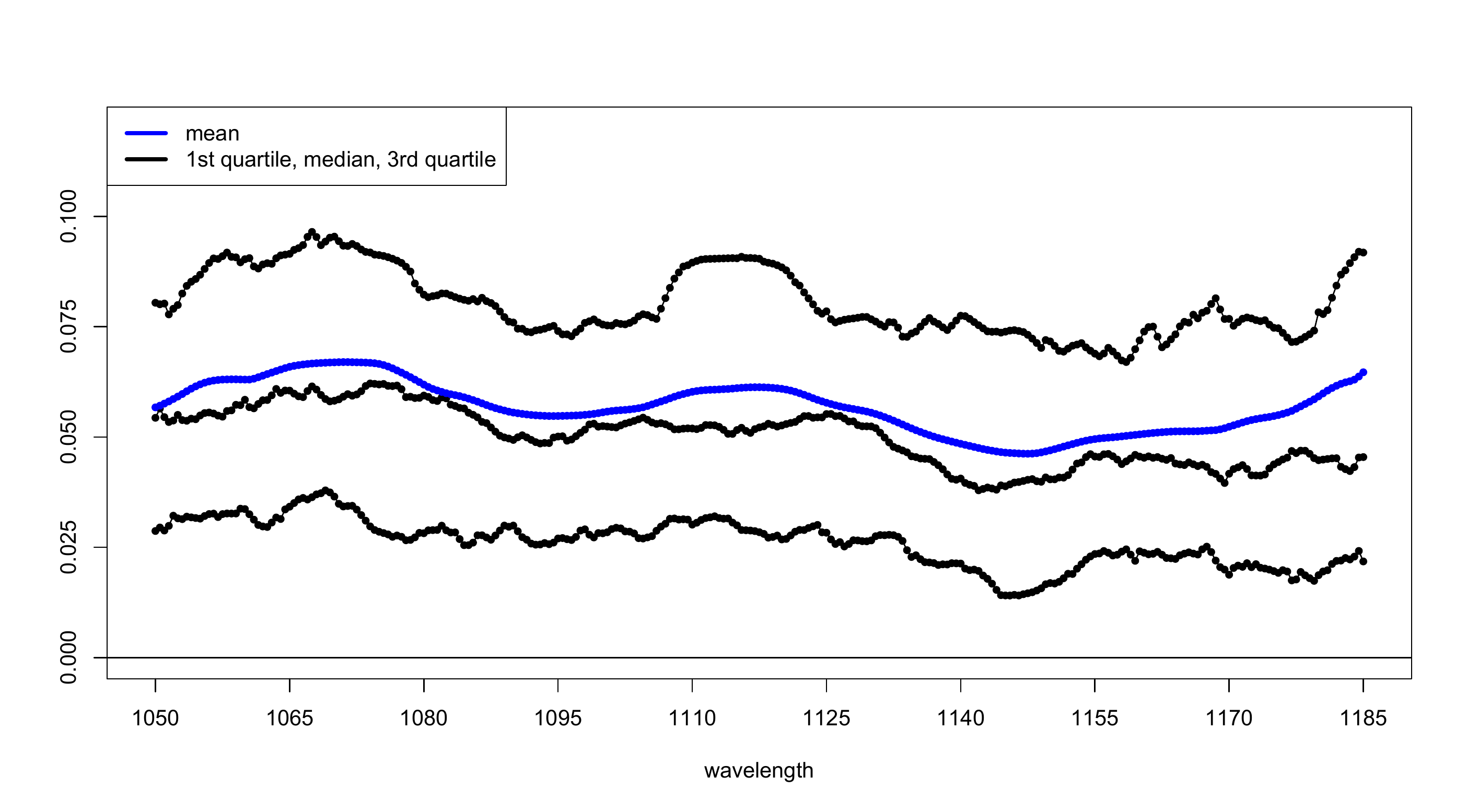}
		\caption{Mean, median, 1st quartile and 3rd quartile of the empirical relative prediction error $e$ computed from the predictions on the 100 mock spectra. The relative prediction error is approximately equal to 5.7\% on average over the target wavelength range.}
		\label{fig:mockpredictionerror}
	\end{center}
\end{figure}

\begin{figure}[h]
	\begin{center}
		\includegraphics[width=\columnwidth]{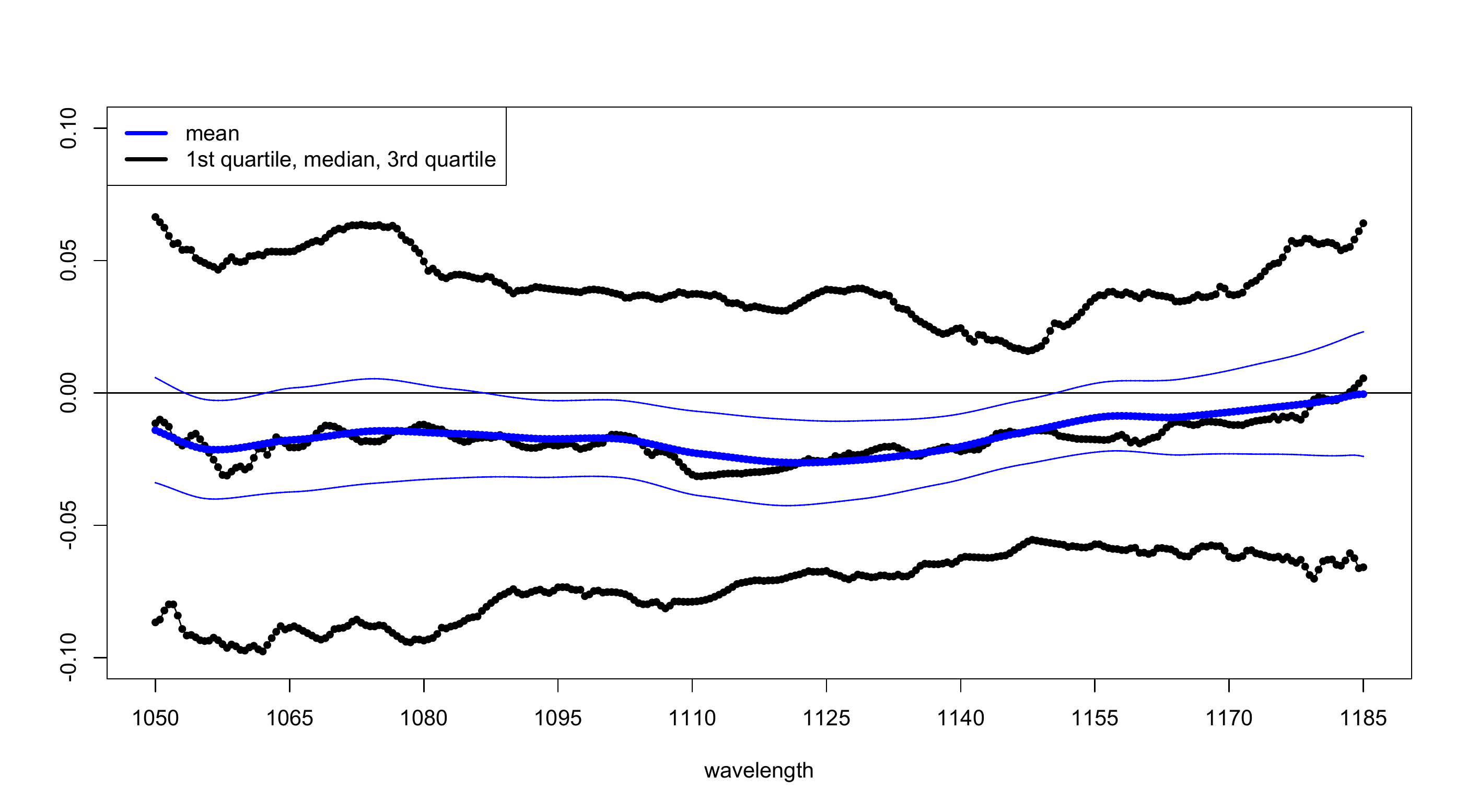}
		\caption{Mean, median, 1st quartile and 3rd quartile of the empirical prediction error $u$ computed from the predictions on the 100 mock spectra. The thin blue lines represent approximate 95\% pointwise confidence intervals for the mean of $u$.}
		\label{fig:mockerror}
	\end{center}
\end{figure}

\section{Application to HST-FOS and BOSS spectra}
\label{sec:realspectraapplication}
In this section we present the prediction of UFC$\alpha$ on real data. More precisely, we perform the following two separate analyses:

\begin{enumerate}
	\item fit the model using 50 of 67 spectra with redshift $z \approx 1$ from the Hubble Space Telescope Faint Object Spectrograph catalog (HST-FOS) and predict UFC$\alpha$ on the remaining subset of 17 HST-FOS spectra
	\item fit the model on the entire set of 67 low redshift HST-FOS spectra and use the model to predict UFC$\alpha$ on a set of 10 spectra with redshift $ z \approx 3$ from the BOSS Data Release 9 catalog.
\end{enumerate}

The HST-FOS spectra are low redshift spectra with no or little flux absorption in the Lyman-$\alpha$ forest portion of the spectrum, whereas absorption is clearly present in the Lyman-$\alpha$ forest of the BOSS spectra (see Figure~\ref{fig:BOSSHSTresults}).  Similarly to the mock spectra, all the observed real spectra are smoothed using local quadratic polynomials and normalized at the $\approx 1300$ \AA\, wavelength.

Based on the scree plot of Figure~\ref{fig:pcaHST}, we decide to keep five principal components ($\approx 95\%$ of variance on HST-FOS smooth Lyman-$\alpha$ forest unabsorbed continua) when running the wild bootstrap. Thus, the bands in the bottom panels of Figure~\ref{fig:BOSSHSTresults} are uniform (with respect to wavelength) approximate confidence bands for $\tilde r^5(x)$, the projection of the unknown regression operator on the first five random basis functions associated to 67 the HST-FOS response functions.

\begin{figure}
	\begin{center}
		\includegraphics[width=\columnwidth]{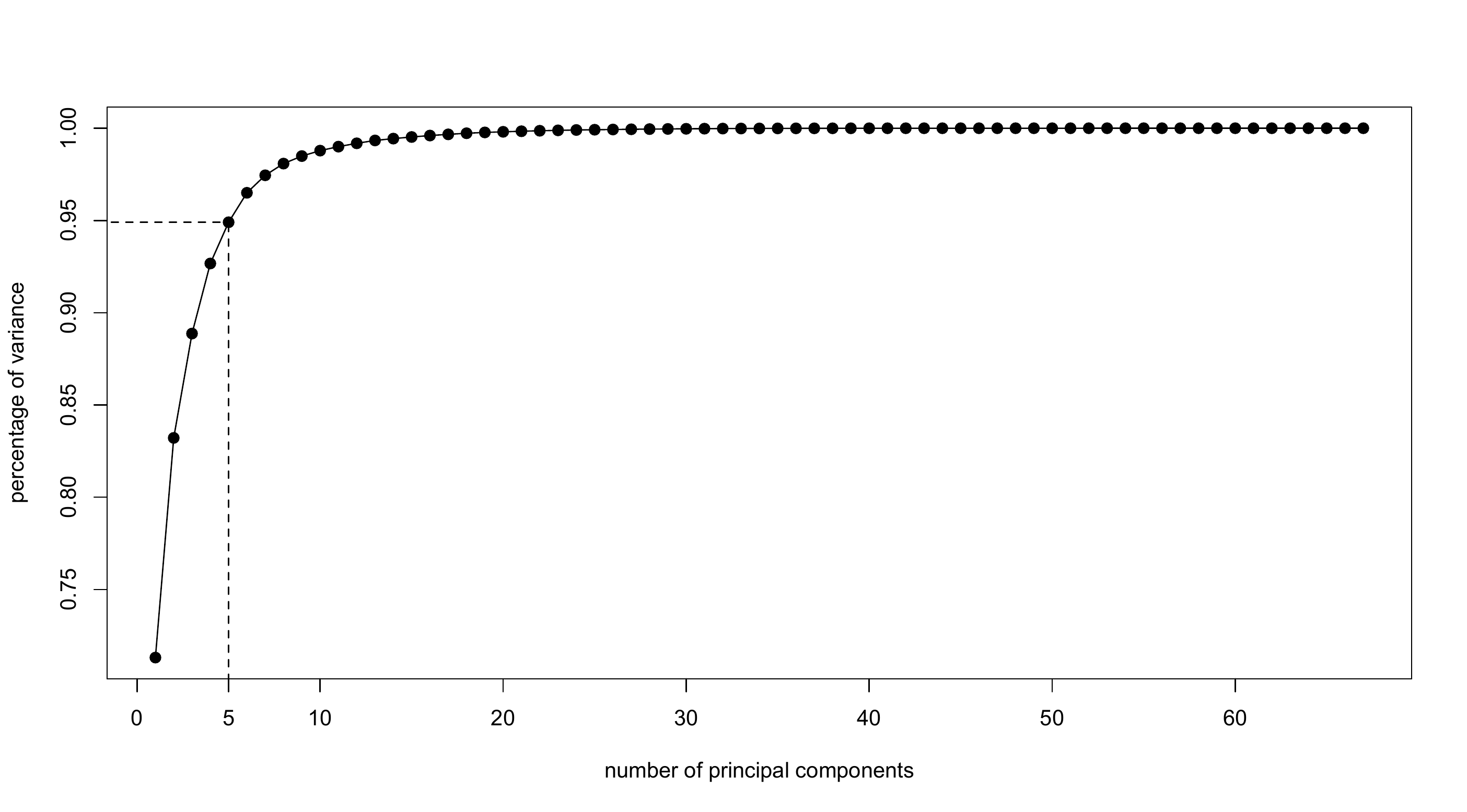}
		\caption{Scree plot of the principal component analysis on the response functions of the 67 HST-FOS spectra. 5 principal components account for $\approx 95\%$ of the variability in the smooth representation of the Lyman-$\alpha$ portion of the HST-FOS spectra.}
		\label{fig:pcaHST}
	\end{center}
\end{figure}

Figure~\ref{fig:HSTHSTresults} displays 4 of the 17 predictions on the HST-FOS spectra when the model is fitted on 50 HST-FOS spectra. Analogously to the results of the previous analysis on the mock spectra, the shape of the flux continuum appears to be correctly predicted: although for real spectra the true UFC$\alpha$ is unknown, the four panels of Figure~\ref{fig:HSTHSTresults} suggest that the predicted continuum follows the same trend of the observed pixels. We point out that, as it was the case with the mock spectra, some of the predicted continua are inaccurate in terms of predicted amplitude (see bottom panels of Figure~\ref{fig:HSTHSTresults}).

\begin{figure}[h]
	\begin{center}
		\includegraphics[width=\columnwidth]{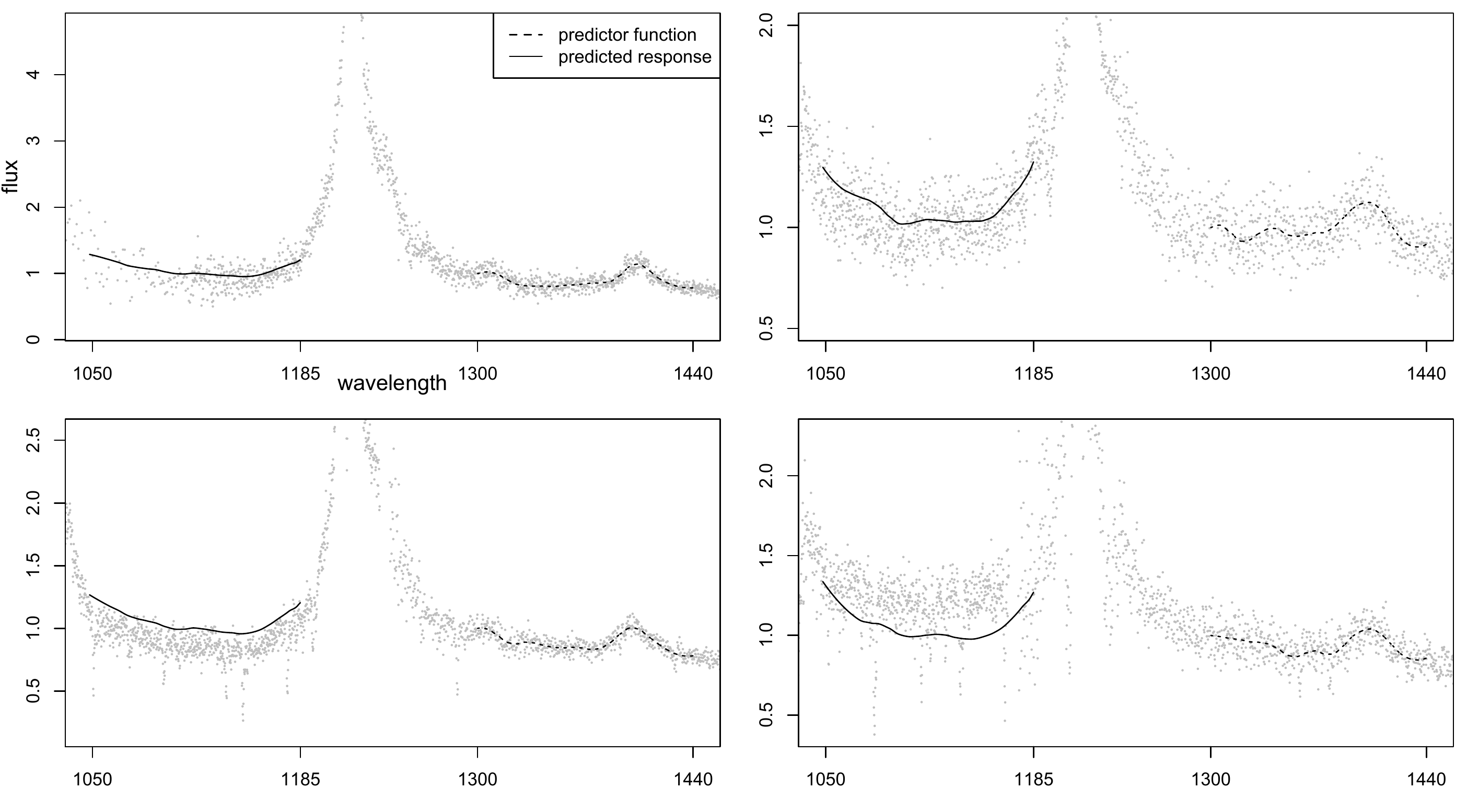}
		\caption{4 of the 17 predictions on HST-FOS spectra when the nonparametric model is fitted on the remaining 50 HST-FOS spectra. \underline{Top panels:} the shape and the amplitude of the flux continuum appears to be correctly predicted. \underline{Bottom panels:} the shape of the continuum is presumably correctly predicted, but the overall amplitude of the underlying flux continuum is not accurately predicted. Notice the presence of absorption in the Lyman-$\alpha$ forest portion of the raw spectra (dotted) in the two bottom panels.}
		\label{fig:HSTHSTresults}
	\end{center}
\end{figure}

Figure~\ref{fig:BOSSHSTresults} displays 4 of the 10 predictions on the BOSS spectra when the model is fitted on the entire set of 67 HST-FOS spectra, together with conformal prediction bands and confidence bands which are derived as described in Sections \ref{sec:conformalsets} and \ref{sec:confidencebands}. The model seems to perform satisfactorily on the 10 BOSS spectra: both the shape and the amplitude of the flux continuum appear to be correctly predicted on the basis of a visual inspection of the plots.

In the context of the BOSS survey, currently each spectrum in the catalog is associated to two PCA-based predictions of the UFC. The prediction that is judged best according to a criterion which includes visual inspection is reported as the state of the art prediction for that given spectrum. The predictions that one would obtain by applying the nonparametric functional regression model on BOSS spectra would add a third prediction for each spectrum and may lead to improvements to the quality of the continua in the BOSS catalog. This should be verified in future work. As far as one is concerned with BAOs detection, a direct byproduct of the improvement in the quality of the continua is a higher statistical power in the detection of BAOs and an increased accuracy in their measurement.

\begin{figure}[h]
	\begin{center}
		\includegraphics[width=\columnwidth]{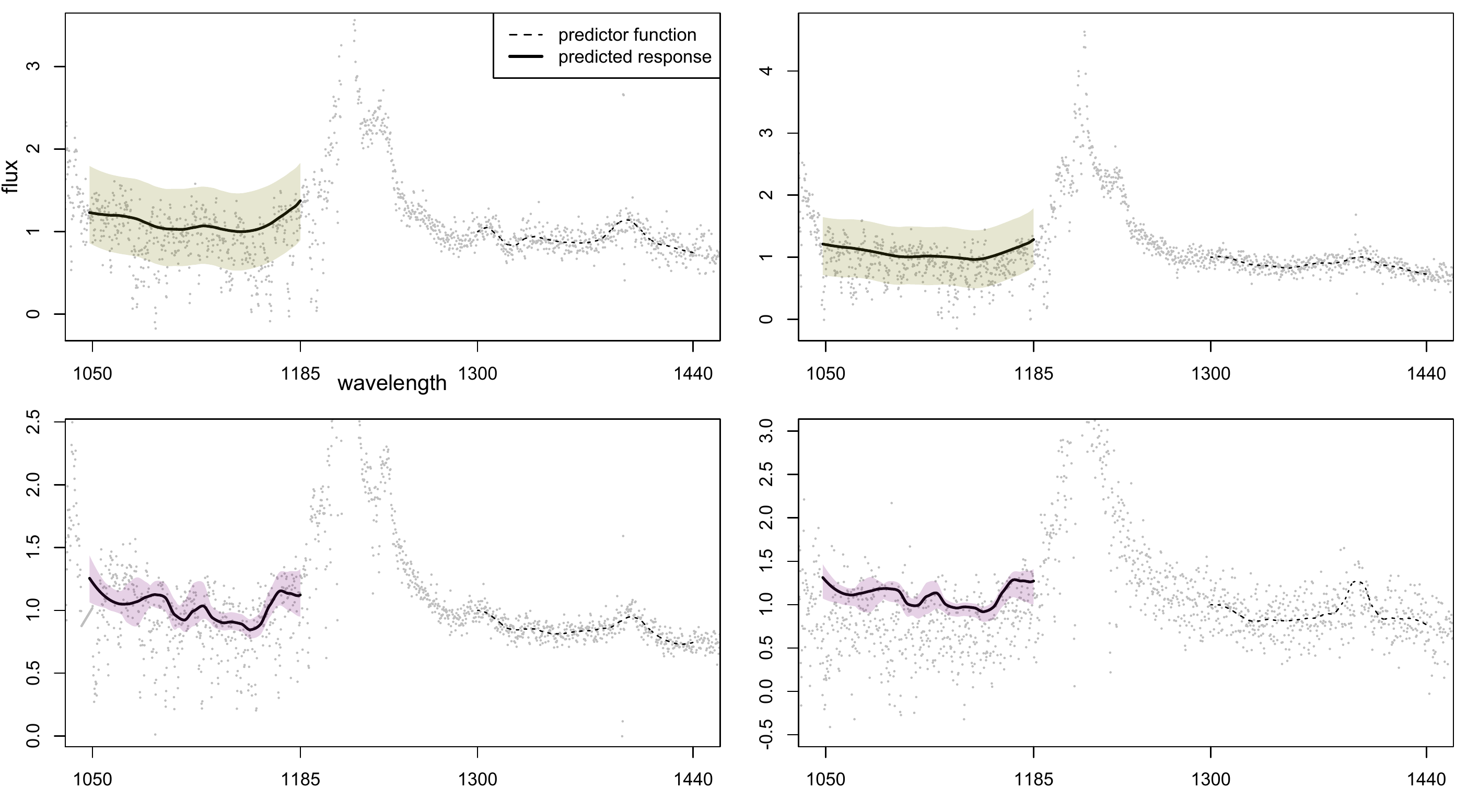}
		\caption{\underline{Top panels:} 2 of the 10 predictions on BOSS quasar spectra with redshift $\approx 3$ when the nonparametric model is fitted on the 67 available HST-FOS quasar spectra with redshift $ z \approx 1$. The shape and the amplitude of the flux continuum appear compatible with the observed flux pixels on the basis of a visual inspection. The gray bands are 90\% marginally valid conformal prediction bands for the true continuum. \underline{Bottom panels:} projection of $\hat r(X)$ onto the first 5 principal components computed on the Lyman-$\alpha$ portion of the HST-FOS spectra. The pink bands in the bottom panels are 90\% wild bootstrap confidence bands for the projection of the unknown regression operator $r(X)$ onto the first $m=5$ principal components of the smooth response functions of the HST-FOS spectra $(Y_i)_{i=1}^{67}$ ($B=500$). Notice that the Lyman-$\alpha$ forest is present and clearly observable in these redshift $z \approx 3$ BOSS spectra.}
		\label{fig:BOSSHSTresults}
	\end{center}
\end{figure}

\section{Conclusions and future work}
\label{sec:conclusions}
In this paper, we demonstrate how the nonparametric functional regression model for function-valued predictor and function-valued response and the estimator of equation \eqref{eq:model} can be used to predict the unknown flux continuum in the Lyman-$\alpha$ forest UFC$\alpha$ from the flux continuum UFC+ in the absorption-free portion of the spectrum using a set of realistic mock spectra and a set of real spectra from the HST-FOS and from the BOSS catalogs. Our results suggest that the methodology that we describe can expand and complement the toolkit of methods that are used to predict the unknown UFC of spectra of large scale studies, which is an important challenge in today cosmology. We hope this study will draw the attention of other astronomers, astrostatisticians and scientists working with Lyman-$\alpha$ forest data to the potential of functional modeling. Furthermore, we introduce a methodology to construct prediction bands with finite sample coverage guarantees and no assumption on the distribution of the function-valued pairs $(X_i,Y_i)$'s for the nonparametric functional regression model proposed by \cite{ferraty2012regression}. It is worth pointing out that current estimates of the correlation function of the relative flux absorption $\delta$ of equation \eqref{eq:relativefluxabsorption} use weighting schemes of the spectra that generally do not account for the uncertainty in the prediction of UFC$\alpha$. The conformal prediction bands proposed in this paper offer a starting point to improve the weighing of the spectra in the computation of the correlation function by including some measure of uncertainty (e.g. the width of the bands) along with the other sources of uncertainty that are already taken into consideration.

Our study raises additional questions, both applied and theoretical. On the applied side, 
it would be interesting to examine if and how strongly the mean flux regulation post-processing step proposed by \cite{lee2012mean} would positively affect the accuracy of the predictions obtained using the nonparametric functional regression model that we describe (especially on BOSS spectra). From a methodological perspective, efforts should be made to devise an extended model that incorporates information about the break in the power law of the spectrum, which we mention in Section \ref{sec:simulationstudy}. Furthermore, information contained in the observed pixels of the Lyman-$\alpha$ forest about the amplitude of UFC$\alpha$ and its shape should be incorporated in the model in such a way to mitigate the inaccuracy in the predicted amplitude of the continuum and to reduce the potential for the slight bias in the predictions which we observed in the simulation study of Section \ref{sec:simulationstudy} (perhaps thus permitting to get around the need for post-processing). Finally, we will investigate refinements of both the conformal prediction bands and of the confidence bands in a future paper.

\bibliography{MyBibliography.bib}

\end{document}